\documentclass[a4paper,10pt]{article}
\usepackage{amssymb,amsmath,amsthm,mathrsfs}
\usepackage{fullpage}
\usepackage{eucal}
\usepackage{color}
\usepackage{stackrel}
\usepackage{graphicx}
\newcommand{\Real}{\mathbb{R}}

\newcommand{\Nat}{\mathbb{N}}

\newcommand{\supp}{\mathop{\mathrm{supp}}\nolimits}
\newcommand{\Dom}{\mathsf{D}}

\newcommand{\eps}{\varepsilon}
\newcommand{\sii}{L^2}

\newcommand{\der}{\mathrm{d}}
\newcommand{\demi}{\mbox{$\frac{1}{2}$}}
\newtheorem{Theorem}{Theorem}

\newtheorem{Proposition}{Proposition}

\theoremstyle{definition}

\numberwithin{equation}{section}
\begin{document}
%
\title{\textbf{\LARGE
Absolute continuity of the spectrum
in a twisted {D}irichlet-{N}eumann waveguide
}}
\author{Ph. Briet,$^{1}$ \ J. Dittrich$^{2}$ \ and \ D. Krej\v{c}i\v{r}\'{\i}k$^{3}$
\bigskip \\
{\small $^1$ Aix-Marseille Universit\'{e}, Universit\'{e} de Toulon, CNRS, CPT, F-13288 Marseille, France}
\\
{\small $^2$ Nuclear Physics Institute, Czech Academy of Sciences, CZ-250 68 \v{R}e\v{z}, Czech Republic}
\\
{\small$^3$ Department of Mathematics, Faculty of Nuclear Sciences and Physical Engineering, }
\\
{\small Czech Technical University in Prague, Trojanova 13, CZ-120 00 Prague 2, Czech Republic}
\\
{\small Electronic mail: briet@univ-tln.fr, dittrich@ujf.cas.cz, david.krejcirik@fjfi.cvut.cz}
}
\date{\small 20 January 2020}
\maketitle
\abstract{Quantum waveguides with the shape of a planar infinite straight strip and combined Dirichlet and Neumann boundary conditions on the opposite
half-lines of the boundary are considered. The absence of the point as well as of the singular continuous  spectrum is proved.}

\section{Introduction}
%
Two-dimensional straight waveguides with combined boundary conditions, classical as well as quantum, were considered in a number of papers
\cite{ELV}--\cite{BC}. Mostly, the existence of isolated eigenvalues was studied. We consider a very special configuration of such quantum waveguides here
for which we show the absence of the eigenvalues, including that embedded in the essential spectrum ones, and the absence of the singular continuous spectrum.

Let~$H$ be the operator that acts as the Laplacian
in a straight strip $\Omega := \Real \times (0,d)$ with $d>0$
and satisfies Dirichlet boundary conditions
on $\partial_D\Omega := [(-\infty,0) \times \{0\}] \cup [(0,\infty)\times\{d\}]$
and Neumann boundary conditions on the other part of the boundary
$\partial_N\Omega := [(-\infty,0) \times \{d\}] \cup [(0,\infty)\times\{0\}]$.
We understand~$H$ as the self-adjoint operator in the Hilbert space $\sii(\Omega)$
generated by the closed form
\begin{equation}\label{form}
  h[\psi] := \int_\Omega |\nabla\psi|^2
  \,, \qquad
  \Dom(h) := \{\psi \in H^1(\Omega)| \
  \psi\upharpoonright\partial_D\Omega = 0\}
  \,.
\end{equation}
One has
\begin{equation*}
  H\psi = -\Delta\psi
  \,, \qquad
  \Dom(H) = \{\psi \in H^1(\Omega) \,| \
  \Delta\psi\in\sii(\Omega) \,, \
  \psi\upharpoonright\partial_D\Omega = 0 \,, \
  \partial_y\psi\upharpoonright\partial_N\Omega = 0
  \}
  \,.
\end{equation*}
Here we denote by $(x,y)$ a generic point in~$\Omega$.

The model belongs to the configurations introduced in~\cite{DKriz1}.
Let $E_n := (2n-1)^2\pi^2/(2d)^2$ with $n\in\Nat^*:=\Nat\setminus\{0\}$
denote the eigenvalues of the Laplacian in $\sii((0,d))$,
subject to a Dirichlet boundary condition at~$0$
and a Neumann boundary condition at~$d$ (or vice versa).
It is easy to see that
\begin{equation*}
  \sigma(H) = \sigma_\mathrm{ess}(H) = [E_1,\infty)
  \,.
\end{equation*}
In~\cite{KK2} it was shown that the operator~$H$ satisfies
a Hardy-type inequality $H-E_1 \geq c/(1+x^2)$
with a positive constant~$c$,
and in~\cite{KZ2}, the consequences on the behavior
of the heat semigroup $e^{-tH}$
for large times $t>0$ were studied.
In particular, it follows that~$E_1$ cannot be an eigenvalue of~$H$.
As the last progress,
the existence of a scattering stationary wave function was established
in \cite{Briet-Dittrich-Soccorsi_2014}.

To complete the study of the model, in this paper,
we study the nature of the essential spectrum and
show that the spectrum of~$H$ is actually purely absolutely continuous.
\begin{Theorem}\label{Thm.main}
One has
$$
  \sigma_\mathrm{p}(H) = \varnothing
  \qquad \mbox{and} \qquad
  \sigma_\mathrm{sc}(H) = \varnothing
  \,.
$$
\end{Theorem}

The idea of our proof of the absence of the point spectrum is based on
the (here formal) commutator identity
\begin{equation}\label{commutator}
  i[H,A] = -2 \, \partial_x^2
  \,,
\end{equation}
where~$A$ is the dilation operator in the longitudinal direction
acting as
\begin{equation}\label{conjugate}
  A := -\frac{i}{2} \, (x \, \partial_x + \partial_x \, x)
  \,.
\end{equation}
It follows from~\eqref{commutator} that
if there exists $u \in \Dom(H)\cap\Dom(A)$ such that $Hu=\lambda u$
with $\lambda \in \Real$, then
\begin{equation*}
  0 = (u,i[H,A]u) = 2 \|\partial_x u\|^2
  \,,
\end{equation*}
where $(\cdot,\cdot)$ and $\|\cdot\|$ denote
the inner product and norm in $\sii(\Omega)$, respectively.
Consequently, $\partial_x u = 0$ as an element of $\sii(\Omega)$,
and therefore, necessarily $u=0$.
It essentially shows that the point spectrum of~$H$ is empty.
To prove the other statement of Theorem~\ref{Thm.main},
we employ the positivity of the right-hand side of~\eqref{commutator},
apart from the set of thresholds
\begin{equation}\label{thresholds}
  \mathcal{T} := \{E_k\}_{k\in\Nat^*}
  \,,
\end{equation}
with help of the Mourre theory of conjugate operators \cite{Mourre}.

The danger of the formal procedure described above is best illustrated
by observing that the same conclusions are obtained for
the modified operator~$H_\eps$ generated by the form~\eqref{form},
where $\partial_D\Omega$ is replaced by
$\partial_D^\eps\Omega := [(-\infty,-\eps) \times \{0\}]
\cup [(\eps,\infty)\times\{d\}]$ with any real~$\eps$.
But if~$\eps$ is positive (so that the Neumann boundary conditions overlap)
and sufficiently large, then it is known (see~\cite{DKriz1})
that~$H_\eps$ admits (discrete) eigenvalues.
The reason behind this apparent contradiction is the fact that
the function~$Au$ does not necessarily belong to~$\Dom(H)$,
so identity~\eqref{commutator} does not make sense
even when applied to $u \in \Dom(H)$.

We prove the absence of the point and singular continuous spectrum for a very special configuration of the planar straight
quantum waveguide with combined Dirichlet and Neumann boundary conditions. While the specific configuration is essential for
the non-existence of discrete eigenvalues,
the absence of the singular continuous spectrum is a more robust property.
As the used conjugate
operator is localized at infinity (acts as zero near the origin $x=0$), the same proofs can be done for variants of~$H$ modified in a bounded subset of~$\Omega$.
For instance, we could consider an arbitrary finite combination of Dirichlet-Neumann boundary conditions in $(-R,R)\times(0,d)$
or even Robin boundary conditions and perhaps compactly supported potentials. However, the modifications should be such that Proposition \ref{Prop.regularity},
i.e., the bound of $\|\partial_x \psi \| \leq C \|H\psi\|$ used in the estimate of (\ref{fin}), holds. This might be a restriction on the possibility of
the waveguide shape local modifications.

We use the Mourre theory in its original form \cite{Mourre}. More advanced exposition can be found in the book \cite{Amrein}. The first application of the Mourre theory in the context of quantum waveguides is in \cite{Krej_Tiedra_2004};
see also \cite{Tiedra_2006}-\cite{Richard_Tiedra_2016} for further developments.

The organization of the paper is as follows. In order to
justify that the formal argument goes through in our situation $H=H_0$,
in Section~\ref{Sec.p},
we use a cut-off approximation of~$u$ both for large and small~$x$
and proceed by the method of multipliers in the spirit of~\cite{FKV,FKV2}.
It is interesting that this apparently technical regularization
actually gives an insight into why this procedure
for~$H_\eps$ with positive~$\eps$ cannot generally work.
Finally, in Section~\ref{Sec.sc}, we modify~\eqref{conjugate}
to a conjugate operator ``localized at infinity''
and prove a (non-strict) Mourre estimate.

\section{Absence of the point spectrum}\label{Sec.p}
%
Let us assume that there exists an eigenfunction
$u\in \Dom(H) \subset \Dom(h)$
and an eigenvalue $\lambda \in {\mathbb R}$ satisfying
\begin{equation}
\label{eveq}
(H-\lambda)u=0 \, .
\end{equation}
Then for any $v\in \Dom(h)$,
\begin{equation}
\label{testeq}
h(v,u)-\lambda \;\!(v,u) = 0 \, .
\end{equation}
We would like to construct a special $v$ such that from the last equation would follow $u=0$, and so there is no eigenvector.
More precisely, our choice of $v$ would not lie in $\Dom(h)$, so
we need to construct a sequence of regularized functions $v_n \in \Dom(h)$
and obtain the result in the limit.

Without loss of generality, we assume that $u$ is real as $\Re u$ and $\Im u$ satisfy (\ref{eveq}) separately.
As a solution of the differential
equation $-\Delta u - \lambda u=0$, $u\in C^\infty(\Omega)$
(cf., e.g., \cite[Thm.~2.2 of Chapt.~4]{Necas},
together with the Sobolev embedding theorem).
In particular, the derivatives of $u$ and its powers may be calculated as classical.

For the regularization purposes,
let us first define a sequence of functions ($n=2,3,4,\dots$)
\begin{eqnarray}
\label{phin}
\varphi_n(x) := \left\{
\begin{array}{lll}
0 & {\rm for} & x\leq -2n \,, \\
(x+2n)/n & {\rm for} & -2n< x < -n \,, \\
1 & {\rm for} & -n \leq x \leq -n^{-1} \,, \\
n^2(x+n^{-2})/(1-n) & {\rm for} & -n^{-1}<x<-n^{-2} \,, \\
0 & {\rm for} & -n^{-2}\leq x \leq n^{-2} \,, \\
n^2(x-n^{-2})/(n-1) & {\rm for} & n^{-2}<x<n^{-1} \,,\\
1 & {\rm for} & n^{-1} \leq x \leq n \,, \\
(2n-x)/n & {\rm for} & n<x<2n \,,\\
0 & {\rm for} &  x\geq 2n \,,
\end{array}
\right.
\end{eqnarray}
belonging to $H^1(\Real)$ with the derivatives $\varphi_n'$ defined almost everywhere. Then set
\begin{equation}
v_n(x,y) := \varphi_n(x) (2xu_x(x,y)+u(x,y)) \, .
\end{equation}
Now
\begin{align*}
v_{nx} &= \varphi_n'(x)(2xu_x(x,y)+u(x,y)) + \varphi(x) (3 u_x(x,y) + 2 x u_{xx}(x,y)) \,, \\
v_{ny} &= \varphi_n(x)(2xu_{xy}(x,y)+u_y(x,y)) \,.
\end{align*}
Evidently, $v_n\in \Dom(h)$ and so satisfies (\ref{testeq}).
Remembering the properties of $\Dom(H)$ \cite{DKriz1}, $u$,$u_x$,$u_y$, $u_{xx}+u_{yy} \in L^2(\Omega)$, and
$u_{xx},u_{xy},u_{yy}\in L^2({\rm supp\,}\varphi_n \times (0,d))$, we write
\begin{equation}\label{huvn1}
h(v_n,u)
=\int_\Omega \varphi_n'(x)(2x u_x^2 + u u_x)\, \mathrm{d}x \, \mathrm{d}y
 + \int_\Omega \varphi_n(x)(3 u_x^2 + 2 x u_x u_{xx} +2 x u_y u_{xy} + u_y^2)\, \mathrm{d}x\, \mathrm{d}y \,.
\end{equation}
Integration by parts with respect to $x$, and also with respect to $y$ in one case, gives
\begin{align*}
  \int_\Omega \varphi_n' u u_x \, \mathrm{d}x \, \mathrm{d}y
  &= - \int_\Omega \varphi_n (u_x^2 + u u_{xx}) \, \mathrm{d}x \, \mathrm{d}y \\
  &= - \int_\Omega \varphi_n (u_x^2 +u \Delta u -u u_{yy})\, \mathrm{d}x \, \mathrm{d}y \\
  &= - \int_\Omega \varphi_n (u_x^2 + u_y^2 + u \Delta u)\, \mathrm{d}x \, \mathrm{d}y \,,
\\
  \int_\Omega \varphi_n (2 x u_x u_{xx} + 2 x u_y u_{xy}) \, \mathrm{d}x \, \mathrm{d}y
  &= \int_\Omega \varphi_n x (u_x^2 + u_y^2)_x \, \mathrm{d}x \, \mathrm{d}y \\
  &= -\int_\Omega \varphi_n' x (u_x^2 + u_y^2) \, \mathrm{d}x \, \mathrm{d}y
  - \int_\Omega \varphi_n (u_x^2 + u_y^2) \, \mathrm{d}x \, \mathrm{d}y \,.
\end{align*}
Inserting to (\ref{huvn1}), we get $h(v_n, u)=I_n + J_n$ with
\begin{equation*}
I_n := \int_\Omega \varphi_n'(x) x ( u_x^2 - u_y^2) \, \mathrm{d}x \, \mathrm{d}y
\,, \qquad
J_n := \int_\Omega \varphi_n(x) (u_x^2-u_y^2 -u \Delta u) \, \mathrm{d}x \, \mathrm{d}y
\,.
\end{equation*}
By similar calculations,
\begin{eqnarray*}
  (v_n,u)
  = \int_\Omega \varphi_n (2 x u u_x + u^2) \, \mathrm{d}x \, \mathrm{d}y
  = \int_\Omega \varphi_n (x (u^2)_x+u^2) \, \mathrm{d}x \, \mathrm{d}y
  = - \int_\Omega \varphi_n'(x) x u^2 \, \mathrm{d}x \, \mathrm{d}y \,.
\end{eqnarray*}

Looking at the definition (\ref{phin}),
it is clear that,
$|\varphi_n|\leq 1$,
$\lim_{n\to \infty} \varphi_n(x) =1$ for every $x\not= 0$
and $|x \varphi_n'(x)| \leq 2$ for every $x \in \Real$.
Furthermore, $\varphi_n'(x)\not = 0$ only for
$$x\in (-2n,-n) \cup (-n^{-1},-n^{-2}) \cup (n^{-2},n^{-1}) \cup (n,2n) \,.$$
Consequently,
\begin{equation*}
\lim_{n\to\infty} I_n = 0 \,,
\qquad
\lim_{n\to\infty} J_n =\int_\Omega (u_x^2-u_y^2-u \Delta u) \, \mathrm{d}x \, \mathrm{d}y
=2 \|u_x\|^2 \,,
\qquad
\lim_{n\to\infty} (v_n,u) = 0 \,,
\end{equation*}
by the dominated convergence.
As
$$
  0=h(v_n,u) - \lambda \;\! (v_n,u)
  = I_n + J_n - \lambda \;\! (v_n,u)
  \xrightarrow[n\to\infty]{} 2 \;\! \|u_x\|^2 \,,
$$
it follows that $u_x=0$, so $u$ is necessarily $x$-independent.
Now $u=0$ because $u\in L^2(\Omega)$, and there is no non-zero eigenfunction and no eigenvalue satisfying (\ref{eveq}).
So the relation $\sigma_\mathrm{p}(H) = \varnothing$
from Theorem \ref{Thm.main} is proved.

\section{Absence of the singular continuous spectrum}\label{Sec.sc}
%
Given any $E \in \Real$ and $\delta>0$,
$P_\delta$~will denote the spectral projection of~$H$
onto the interval $(E-\delta,E+\delta)$.
We restrict to $E \not\in \mathcal{T}$,
where the set~$\mathcal{T}$ is introduced in~\eqref{thresholds},
and choose~$\delta$ so small that
$(E-\delta,E+\delta) \cap \mathcal{T} = \varnothing$.
Let~$H$ be as above and let~$A$ be a self-adjoint operator
to be specified in a moment
(it will be a regularization of~\eqref{conjugate}).
To apply the abstract theorem of~\cite{Mourre}
and thus conclude the absence of the singular continuous spectrum of~$H$,
it is enough to verify the following properties:

\begin{enumerate}
\item[(a)]
The intersection $\Dom(A) \cap \Dom(H)$ is a core of $H$.
\item[(b)]
The unitary group $e^{itA}$ leaves the domain of~$H$ invariant and
\begin{equation}\label{invariant}
  \forall \psi \in \Dom(H)
  \,, \qquad
  \sup_{|t| < 1}\|H e^{itA}\psi\| < \infty .
\end{equation}
\item[(c)]
The form
$$
  \dot{b}[\psi] := i (H\psi,A\psi) - i (A\psi,H\psi)
  \,, \qquad
  \Dom(\dot{b}) := \Dom(A) \cap \Dom(H)
  \,,
$$
is bounded from below and closable.
Moreover, the operator~$B$ associated with the closure~$b$ of~$\dot{b}$
satisfies
$$
  \Dom(B) \supset \Dom(H)
  \,.
$$
\item[(d)]
The operator defined by the form
$$
  \dot{c}[\psi] := i (B\psi,A\psi) - i (A\psi,B\psi)
  \,, \qquad
  \Dom(\dot{c}) := \Dom(A) \cap \Dom(H)
  \,,
$$
extends to an operator
$$
  C \in \mathscr{B}(\Dom(H),\Dom(H)^*)
  \,,
$$
$\Dom(H)$ being equipped with the graph norm and $\Dom(H)^*$ being its dual space.
\item[(e)]
There exists a positive number~$\alpha$
and a compact operator~$K$ on~$\sii(\Omega)$
such that
$$
  P_\delta B P_\delta \geq \alpha P_\delta + P_\delta K P_\delta
  \,.
$$
\end{enumerate}

Note that~$B$ (respectively, $C$) can be interpreted as a realization
of the commutator $i[H,A]$
(respectively, the double commutator $i[i[H,A],A]$).

\subsection{The Hamiltonian}
We begin with establishing some new results about the operator~$H$
which will be needed later.
\begin{Proposition}\label{Prop.density}
For every positive~$\epsilon$,
the set
\begin{multline*}
  \mathcal{C} := \big\{ \varphi \in \Dom(H) \,| \
  \exists \phi \in C_0^\infty(\Real^2) \,, \quad
  \\
  \varphi\upharpoonright((-\infty,-\epsilon)\cup (\epsilon,+\infty))\times(0,d)
  =  \phi\upharpoonright((-\infty,-\epsilon)\cup (\epsilon,+\infty))\times(0,d)
  \big\}
\end{multline*}
is a core of~$H$.
\end{Proposition}
\begin{proof}
Let $\psi$ be an arbitrary function from $\Dom(H)$. We show that it can be approximated by functions from $\mathcal{C}$.
Let $\vartheta_1$ and $\vartheta_2$ be functions from $C^\infty(\Real)$ such that $0\leq \vartheta_1, \vartheta_2 \leq 1$ and
\begin{align*}
\vartheta_1(x) &= 1 \quad {\rm for} \quad x<-\epsilon \,,
&\vartheta_1(x) &= 0 \quad {\rm for} \quad x>-\frac{\epsilon}{2} \,,
\\
\vartheta_2(x) &= 1 \quad {\rm for} \quad x>\epsilon \,,
&\vartheta_2(x) &= 0 \quad {\rm for} \quad x<\frac{\epsilon}{2} \,.
\end{align*}
Let us define
$$
\psi_1 = \vartheta_1 \psi \,,
\qquad \psi_2 = \vartheta_2 \psi \,,
\qquad \psi_3 = (1-\vartheta_1 -\vartheta_2)\psi \,,
$$
so that
$$
\psi = \psi_1 + \psi_2 +\psi_3 \qquad {\rm and} \qquad \psi_1, \psi_2, \psi_3 \in \Dom(H) \,.
$$
It is sufficient to approximate $\psi_1$ and $\psi_2$ by functions from $\mathcal C$. It is known that $\psi_1,\psi_2 \in H^2(\Omega)$; see~\cite{DKriz1}.
Let us extend them to $H^2(\Real \times (-d,2d))$ first. To keep the boundary conditions, let us choose extensions symmetric with respect to the Neumann parts of the boundary and antisymmetric with respect to the Dirichlet parts. Note that in half-planes where the functions are zero, it means the same. So we define
\begin{align*}
\psi_1(x,y)&=-\psi_1(x,-y) \quad {\rm for} \quad -d<y<0 \,,
& \psi_1(x,y)&=\psi_1(x,2d-y) \quad {\rm for} \quad d<y<2d \,,
\\
\psi_2(x,y)&=\psi_2(x,-y) \quad {\rm for} \quad -d<y<0 \,,
& \psi_2(x,y)&=-\psi_2(x,2d-y) \quad {\rm for} \quad d<y<2d \,.
\end{align*}
The extended functions are in $H^2(\Real \times (-d,0))$ and $H^2(\Real \times (d,2d)$). As the traces of functions and the normal derivatives on the boundaries of $\Omega$ from both sides coincide, the extended functions are in $H^2(\Real \times (-d,2d))$. In fact, we used a special case of \cite[Thm~4.26]{Adams} and its proof. Then, extend them to $H^2(\Real^2)$ which is possible over the straight boundary.

Furthermore, we need to approximate $\psi_1$ and $\psi_2$ by $C^\infty$ functions. We use the standard mollifications, see \cite[Lem.~3.15]{Adams},
$$
J_\eta \psi_k(x)=\int_{\Real^2}j_\eta(x-y) \psi_k(y) \, \der y \qquad (k=1,2),
$$
where
$$
j_\eta(x)=\eta^{-2}j(x/\eta) \,,\qquad
j\in C^\infty_0(B(1)) \,,
\qquad j\geq 0 \,,
\qquad \int_{\Real^2}j(x) d^2 x = 1 \,.
$$
Let us consider only $0<\eta<\min(d,\epsilon/2)$
for $\supp j_\eta \subset B(\eta)$. Then, $J_\eta\psi_{1,2} \in H^2(\Real \times (-d,2d))$ and approach $\psi_{1,2}$ there as $\eta \to 0^+$.
These function are in $\Dom(H)$ if they satisfy the corresponding boundary conditions at $\partial \Omega$ which are easily verified for the usual symmetric choice of $j_\eta(x,y)=j_\eta(x,-y)$.

Let us show it here for the case of Neumann boundary condition on $(0,+\infty) \times \{0\}$. The trace exists as $J_\eta\psi_2\in H^2(\Real\times(-d,2d))$
and we can simply calculate
\begin{eqnarray*}
\partial_2 J_\eta\psi_2(x,0)=\int_{\Real^2} \partial_2 j_\eta(x-x',-y')\psi_2(x',y') \; \mathrm{d}x' \; \mathrm{d}y' = \int_{\Real^2} \partial_{2} j_\eta(x-x',-y')\psi_2(x',-y') \; \mathrm{d}x' \; \mathrm{d}y'
\\
= \int_{\Real^2} \partial_2 j_\eta(x-x',y') \psi_2(x',y') \; \mathrm{d}x' \; \mathrm{d}y'
=-\int_{\Real^2} \partial_2 j_\eta(x-x',-y') \psi_2(x',y') \; \mathrm{d}x' \; \mathrm{d}y' = - \partial_2 J_\eta \psi_2 (x,0)
\end{eqnarray*}
and the required boundary condition $\partial_2 J_\eta\psi_2 (x,0)$ at $x>0$ follows. The other boundary conditions are verified similarly.

Finally, let $\Phi_R \in C_0^\infty(\Real^2)$, $\Phi_R(x,y)=\Phi_{1R}(x) \Phi_2(y)$,
where $\Phi_{1R}$ is a suitable function with the support in $(-R-1,R+1)$ and the value $1$ in $(-R,R)$ while $\Phi_2$ is a function with the support in $(-d/2,3d/2)$ and the value $1$ in $(-d/4,5d/4)$.
Then $\phi = \Phi_R ( J_\eta \psi_1 + \psi_3 + J_\eta \psi_2) \in \mathcal{C}$ is an arbitrarily good approximation of $\psi$ in $\Dom(H)$
with the graph norm choosing $\eta$ sufficiently small and $R$ large enough. So ${\mathcal C}$ is a core of $H$.
\end{proof}
\begin{Proposition}\label{Prop.regularity}
There exists a positive constant~$C$ such that,
for every $\psi \in \Dom(H)$,
\begin{equation}\label{bound1}
  \|\partial_x\psi\| \leq C \|H\psi\|
  \,, \qquad
  \|\partial_y\psi\| \leq C \|H\psi\|
  \,.
\end{equation}
Moreover, for every positive~$\epsilon$,
there exists a positive constant~$C_\epsilon$ such that,
for every $\psi \in \Dom(H)$,
\begin{equation}\label{bound2}
  \|\chi_\eps\partial_x^2\psi\| \leq C_\epsilon \|H\psi\|
  \,,
  \qquad
  \|\chi_\eps\partial_x\partial_y\psi\| \leq C_\epsilon \|H\psi\|
  \,.
\end{equation}
where $\chi_\epsilon$ denotes the characteristic function
of the set $\Omega \setminus [(-\epsilon,\epsilon)\times(0,d)]$.
\end{Proposition}
\begin{proof}
Given any $g \in \sii(\Omega)$, let $\psi \in \Dom(H)$
be the unique solution of the resolvent equation $H\psi=g$
[the problem is well defined because $0 \not\in \sigma(H)$].
The weak formulation reads
\begin{equation}\label{weak}
  \forall v \in \Dom(h)
  \,, \qquad
  (\partial_x v,\partial_x \psi) + (\partial_y v,\partial_y \psi)
  = (v,g)
  \,.
\end{equation}
Choosing $v:=\psi$ in~\eqref{weak}, we get
\begin{equation*}
  E_1 \|\psi\|^2 \leq
  \|\partial_x \psi\|^2 + \|\partial_y \psi\|^2
  = (\psi,g)
  \leq \|\psi\| \|g\|
  \,.
\end{equation*}
Consequently, $\|\psi\| \leq E_1^{-1} \|g\|$,
$\|\partial_x\psi\|^2 \leq E_1^{-1} \|g\|^2$
and $\|\partial_y\psi\|^2 \leq E_1^{-1} \|g\|^2$.
This proves~\eqref{bound1}.

To establish~\eqref{bound2}, we follow the ideas
of standard elliptic regularity (see \cite[Sec.~6.3]{Evans}).
Let $\xi \in C_0^\infty(\Real)$ be such that $0 \leq \xi \leq 1$,
$\xi(x)=0$ if $|x|\leq \epsilon/2$ and $\xi(x)=1$ if $|x|\geq \epsilon$.
Now we choose
$v:=-\partial_x^{-h}(\xi^2\partial_x^h \psi)$ in~\eqref{weak},
where
\begin{equation*}
  \partial_x^h\varphi(x,y) := \frac{\varphi(x+h,y)-\varphi(x,y)}{h}
  \,,
\end{equation*}
is the difference quotient of $\varphi \in \sii(\Omega)$ in the direction~$x$.
With an abuse of notation (followed also at other places in the paper),
we denote by the same symbol~$\xi$ the function on~$\Real$
as well as $\xi \otimes 1$ on~$\Omega$.
Choosing $|h| \leq \epsilon/2$, we have $v \in \Dom(h)$
(it is only important to ensure the Dirichlet boundary conditions).
Using the integration-by-parts formula for the difference quotients,
\eqref{weak} yields
\begin{equation}\label{id.reg}
 | \, \|\xi\partial_x^h\partial_x \psi\|^2
  + 2 \, (\xi'\partial_x^h\psi,\xi\partial_x^h\partial_x\psi)
  + \|\xi\partial_x^h\partial_y \psi\|^2 |
  = |(v,g)|
  \leq \|v\| \|g\|
  \,.
\end{equation}
To deal with the right-hand side, we write
\begin{equation*}
  \|v\|^2
  = \|\partial_x^{-h}(\xi^2\partial_x^h \psi)\|^2
  \leq \|\partial_x(\xi^2\partial_x^h \psi)\|^2
  \leq 2 \|\xi^2\partial_x^h\partial_x \psi\|^2
  + 2 k_\epsilon^2 \|\partial_x^h\psi\|^2
  \leq 2 \|\xi\partial_x^h\partial_x \psi\|^2
  + 2 k_\epsilon^2 \|\partial_x\psi\|^2
  \,,
\end{equation*}
where $\|(\xi^2)'\|_\infty \leq 2 \|\xi'\|_\infty =: k_\epsilon$.
On the left-hand side, we use
\begin{equation*}
  2 \, |(\xi'\partial_x^h\psi,\xi\partial_x^h\partial_x\psi)|
  \leq 2 \|\xi'\partial_x^h\psi\| \|\xi\partial_x^h\partial_x\psi\|
  \leq k_\epsilon \|\partial_x\psi\| \|\xi\partial_x^h\partial_x\psi\|
  \,.
\end{equation*}
Consequently, \eqref{id.reg}~yields
\begin{equation*}
\begin{aligned}
  (1 - \delta_1 - 2 \delta_2) \|\xi\partial_x^h\partial_x \psi\|^2
  + \|\xi\partial_x^h\partial_y \psi\|^2
  &\leq k_\epsilon^2 \left(\frac{1}{\delta_1} + 2 \delta_2\right) \|\partial_x\psi\|^2
  + \frac{1}{\delta_2} \|g\|^2
  \\
  &\leq \left[
  k_\epsilon^2 \left(\frac{1}{\delta_1} + 2 \delta_2\right) E_1^{-1}
  + \frac{1}{\delta_2} \right]
  \|g\|^2
\end{aligned}
\end{equation*}
with any positive numbers~$\delta_1$ and~$\delta_2$,
where the second inequality employs~\eqref{bound1}
with the explicitly given constant.
Choosing~$\delta_1$ and~$\delta_2$ sufficiently small,
the left-hand side is a sum of two non-negative terms
and the desired claims follow
by further estimating
$
  \|\xi\partial_x^h\partial_x \psi\|^2
  \geq \|\chi_\epsilon\partial_x^h\partial_x \psi\|
$
(and similarly for the other norm)
and by sending~$h$ to~$0$.
\end{proof}

\subsection{The conjugate operator}
\label{conjugate.operator}
Let $f_1^\pm \in C^\infty(\Real)$ be such that $0 \leq f_1^\pm \leq 1$,
$f_1^\pm(x)=0$ if $\pm x \leq 1$ and $f_1^\pm(x)=1$ if $\pm x \geq 2$.
For every $n \geq 1$, we define $f_n^\pm(x) := f_1^\pm(x/n)$
and $F_n^\pm(x) := \int_0^x f_n^\pm(\xi) \, \der\xi$.
Finally, we set $f_n := f_n^- + f_n^+$ and $F_n := F_n^- + F_n^+$.
Notice that $F_n^\pm(x) \sim x$ as $x \to \pm\infty$
and that $\|(f_n^\pm)^{(m)}\|_\infty=n^{-m}\|(f_1^\pm)^{(m)}\|_\infty$.

With these preliminaries, we define
\begin{equation}\label{conjugate.parallel}
  \dot{A}_\parallel
  := -\frac{i}{2} \, \big(F_n(x) \, \partial_x + \partial_x \, F_n(x)\big)
  \,, \qquad
  \Dom(\dot{A}_\parallel) := C_0^\infty(\Real)
  \,,
\end{equation}
where~$F_n$ is understood as an operator of multiplication. The following considerations are full analogy of \cite[Props.~6.1--2]{Briet-Kovarik-Raikov_2014}.
However, as there is a difference in the cut-off at zero
instead of the cut-off at infinity,
we give the proofs here.

The operator $\dot{A}_\parallel$ is essentially self-adjoint in $L^2(\Real)$. This is a consequence of \cite[Prop.~7.3.6(a)]{Amrein} and its proof. In our special case,
it can also be seen directly that the deficiency indices of $\dot{A}_\parallel$ are zero due to the properties of function $F_n$.

Let~$A_\parallel$ denote the (self-adjoint) closure of~$\dot{A}_\parallel$.
Using the Hilbert-space identification
$\sii(\Omega) \cong \sii(\Real) \times \sii((0,d))$,
we set
\begin{equation}\label{conjugate.tensor}
  A := A_\parallel \otimes 1
  \,,
\end{equation}
which is a self-adjoint operator in $\sii(\Omega)$.

For any fixed $x \in \Real$,
consider the initial-value problem
\begin{equation}\label{initial}
\left\{
\begin{aligned}
  \frac{\der}{\der t} u(t,x) &= F_n(u(t,x)) \,,
  \\
  u(0,x) &= x \,.
\end{aligned}
\right.
\end{equation}
By classical results (see \cite[Thm.~4.1 of Chapt.~V]{Hartman}),
\eqref{initial}~admits a unique global solution in $C^\infty(\Real^2)$.
One has
\begin{equation}\label{wt0}
  \partial_x u (t,x) = e^{\int_0^t f_n(u(s,x))\, \der s} >0
\end{equation}
for every $t \in \Real$ and $x \in \Real$.
Define
\begin{equation}\label{wt}
  (W(t)\varphi )(x,y) := |\partial_xu(t,x)|^{1/2}\, \varphi(u(t,x), y)
  \,.
\end{equation}
\begin{Proposition} \label{Unitary_group}
$W$ is a strongly continuous unitary group on $\sii(\Omega)$ with
the generator \eqref{conjugate.tensor}.
\end{Proposition}
\begin{proof}
It is clear from (\ref{initial}) that $u(t,0)=0$ for $t\in \Real$, and $u(t,x)\gtrless 0$ for $x\gtrless 0$. Using the properties of $f_n$,
the relation (\ref{wt0}) is now improved to
$$
\partial_x u (t,x) \geq e^{-|t|}
$$
for every $t,x \in \Real$ and
$$
\lim_{x\to \pm \infty}u(t,x) = \pm \infty .
$$
The unitarity of $W(t)$ then follows from its construction (\ref{wt}).

Equation (\ref{initial}) together with the unicity of its solution implies the relation
$$
u(t,u(s,x))=u(t+s,x) \,,
$$
from which the group property
$$
W(t)W(s)=W(t+s)
$$
follows.

It is sufficient to verify the strong continuity of $W(t)$ at $t=0$.
The continuity of $W(t)\varphi$ is easily seen for
$\varphi \in C_0^\infty (\Omega)$
and then extends to $\varphi \in \sii (\Omega )$
by the density argument as $\| W(t)\|=1$.

Direct calculations show
$$
\frac{d}{dt} W(t) \varphi_{|t=0} =  i (\dot{A}_\parallel \otimes 1) \varphi
$$
for $\varphi \in C^\infty_0 (\Omega)$. As the generator of the group $W$ is self-adjoint, it equals $A$ necessarily.
\end{proof}

The following proposition establishes property~(b).
\begin{Proposition}
$\Dom(H)$  is stable under the action of $e^{it A}$ and~\eqref{invariant} holds.
\end{Proposition}
\begin{proof}
Let $\varphi \in \Dom(H)$. We need to check that then $e^{itA} \varphi = W(t)\varphi \in \Dom(H)$, for every $t\in\Real$.
We have seen in the previous proof that the map $ \Real \ni x \mapsto u(t,x) \in \mathbb R$
leaves  $ \mathbb R ^+$ and $\mathbb R^-$ invariant. So $e^{itA} \varphi$ satisfies the required boundary conditions
at $\partial_D\Omega$ and $\partial_N\Omega$.

Equation (\ref{wt0}) implies that the derivatives $\partial_x u$, $\partial_x^2 u$, and $\partial_x^3 u$ are bounded in $x$ for a fixed $t$.
Then $e^{itA} \varphi \in H^1(\Omega)$. Let us calculate
\begin{eqnarray*}
\Delta e^{itA}\varphi = W(t)\Delta\varphi + (\partial_x u)^{\frac{1}{2}}((\partial_x u)^2 -1) \partial_1^2 \varphi(u,y) + 2 (\partial_x u)^{\frac{1}{2}} (\partial_x^2 u) \partial_1\varphi(u,y)
\\
+(\partial_x u)^{\frac{1}{2}} \left(\frac{1}{2} (\partial_x u)^{-1} \partial_x^3 u - \frac{1}{4}(\partial_x u)^{-2} (\partial_x^2 u)^2 \right)\varphi(u,y)
\,.
\end{eqnarray*}
Every terms on the right-hand side are clearly square integrable, possibly  except of the second one. However, $\partial_x u(t,x)=1$ for $|x| < e^{-|t|} n$ according to (\ref{wt0}) and the properties of $f_n$.
So the second term is also square integrable as $\partial_1^2 \varphi \in L^2(\Omega \setminus ((u(t,-e^{-|t|}n),u(t, e^{-|t|}n))\times (0,d)))$, see~\cite{DKriz1}.
Now the relation $e^{itA}\varphi \in D(H)$ is proved.
Further, the continuity of the used bounds with respect to~$t$ implies~\eqref{invariant}.
\end{proof}

The following proposition establishes property~(a).
\begin{Proposition}
$\Dom(A) \cap \Dom(H)$  is dense in $\Dom(H)$ for the graph norm associated with~$H$.
\end{Proposition}
\begin{proof}
The claim follows from Proposition~\ref{Prop.density}
and the fact that $\mathcal{C} \subset \Dom(A)$.
\end{proof}

\subsection{The first commutator}
\label{first.commutator}
Let $\psi \in \Dom(A) \cap \Dom(H)$.
Using the formula~\eqref{conjugate.tensor} with~\eqref{conjugate.parallel}
and integrating by parts, we compute
\begin{equation*}
\begin{aligned}
  \dot{b}[\psi]
  &= 2\Re (-\partial_x^2\psi-\partial_y^2\psi,
  F_n\partial_x\psi+\demi F_n'\psi)
  \\
  &= - \int_\Omega F_n \partial_x|\partial_x\psi|^2
  - \Re \int_\Omega F_n' \overline{\partial_x^2\psi} \psi
  - 2 \Re \int_\Omega F_n \overline{\partial_y^2\psi} \partial_x\psi
  - \Re \int_\Omega F_n' \overline{\partial_y^2\psi} \psi
  \\
  &= \int_\Omega F_n' |\partial_x\psi|^2
  + \int_\Omega F_n' |\partial_x\psi|^2
  + \frac{1}{2} \int_\Omega F_n'' \partial_x|\psi|^2
  + \int_\Omega F_n \partial_x |\partial_y\psi|^2
  + \int_\Omega F_n' |\partial_y\psi|^2
  \\
  &= 2 \int_\Omega F_n' |\partial_x\psi|^2
  - \frac{1}{2} \int_\Omega F_n''' |\psi|^2
  \\
  &= 2 \int_\Omega f_n |\partial_x\psi|^2
  - \frac{1}{2} \int_\Omega f_n'' |\psi|^2
  \,,
\end{aligned}
\end{equation*}
keeping in mind the properties of $F_n$ and $\psi \in \Dom(A) \cap \Dom(H)$.
For brevity, here we have stopped to write the measures of integration
in the integrals.

Since~$f_n$ is non-negative,
we immediately see that~$\dot{b}$ is bounded from below.
Explicitly,
\begin{equation*}
  \dot{b}
  \geq -\frac{\|f_n''\|_\infty}{2}
  = -\frac{\|f_1''\|_\infty}{2n^2}
  \,,
\end{equation*}
so the lower bound actually tends to~$0$ as $n \to \infty$.

Since $\dot{b}[\psi] = (\psi,\dot{B}\psi)$, where
\begin{equation*}
  \dot{B} := -2 \partial_x f_n(x) \partial_x - \frac{1}{2} f_n''(x)
  \,, \qquad
  \Dom(\dot{B}) := \Dom(A) \cap \Dom(H)
  \,,
\end{equation*}
is obviously a symmetric below bounded operator in $\sii(\Omega)$,
it follows that~$\dot{b}$ is closable
(see, \cite[Thm.~VI.1.2.7]{Kato}).
The closure~$b$ satisfies
\begin{equation*}
  b[\psi] = 2 \int_\Omega f_n |\partial_x\psi|^2
  - \frac{1}{2} \int_\Omega f_n'' |\psi|^2
  \,, \qquad
  \Dom(b) = \left\{
  \psi \in \sii(\Omega) \, \big| \ \sqrt{f_n} \, \partial_x\psi \in \sii(\Omega)
  \right\}
  \,.
\end{equation*}

By the representation theorem, we have
\begin{equation*}
  B = -2 \partial_x f_n(x) \partial_x - \frac{1}{2} f_n''(x)
  \,, \qquad
  \Dom(B) = \left\{\psi \in \Dom(b) \, \big | \
  \partial_x (f_n \partial_x \psi) \in \sii(\Omega) \right\}
  \,.
\end{equation*}
It is evident that $\Dom(H) \subset \Dom(B)$.

Summing up, in this subsection, we have established property~(c).

\subsection{The second commutator}

Here, we follow the same lines as in the Sec.~\ref{first.commutator}. Let $\psi \in \Dom(A) \cap \Dom(H)$ and compute

\begin{equation*}
  \dot{c}[\psi] =  2\Re (-2 \partial_x f_n(x) \partial_x \psi - \frac{1}{2} f_n''(x)\psi,
  F_n\partial_x\psi+\demi F_n'\psi)
  \,.
\end{equation*}
First, consider
\begin{equation*}
\begin{aligned}
- 4\Re  \int_\Omega  (\partial_x f_n(x) \partial_x \psi)
  F_n \overline{\partial_x\psi} &=- 4   \int_\Omega   f'_n(x) F_n(x)\vert \partial_x \psi \vert^2 -2
 \int_\Omega   f_n(x) F_n(x) \partial_x \vert \partial_x \psi
 \vert ^2  \\
  & =2  \int_\Omega   f^2_n(x) \vert \partial_x \psi \vert^2 - 2
 \int_\Omega   f'_n(x) F_n(x)  \vert \partial_x \psi
 \vert ^2 .
 \end{aligned}
 \end{equation*}
 Then,
 \begin{equation*}
\begin{aligned}
- 2\Re  \int_\Omega  (\partial_x f_n(x) \partial_x \psi)
  F'_n \overline {\psi} &=-    \int_\Omega   f'_n(x) F'_n(x) \partial_x  \vert \psi  \vert^2  -2\Re
 \int_\Omega   f_n(x) F'_n(x) (\partial^2_x \psi)  \overline {\psi}
  \\
  & =  \int_\Omega  f'_n (x)f_n(x)\partial_x   \vert \psi \vert^2 + 2
 \int_\Omega   f^2_n(x)\vert \partial_x \psi
 \vert ^2
 \\
 &= -  \int_\Omega  (f'_n (x)f_n(x))' \vert \psi \vert^2 + 2
 \int_\Omega   f^2_n(x)\vert \partial_x \psi  \vert ^2 .
\end{aligned}
 \end{equation*}
 We also have
 \begin{equation*}
\begin{aligned}
- \Re  \int_\Omega  f''_n(x)F_n( x)  \psi \partial_x  \overline{\psi} = - {\frac{1}{2}}  \int_\Omega  f''_n(x)F_n( x)  \partial_x  \vert \psi \vert^2 =  {\frac{1}{2}}  \int_\Omega  (f''_n(x)F_n( x))'    \vert \psi \vert^2 .
 \end{aligned}
 \end{equation*}

 Finally  we get
   \begin{equation}
\begin{aligned}\label{fin}
 \dot{c}[\psi]= 4  \int_\Omega   f^2_n(x) \vert \partial_x \psi \vert^2 &- 2
 \int_\Omega   f'_n(x) F_n(x)  \vert \partial_x \psi
 \vert ^2   \\
 &- \int_\Omega   \big (f''_n(x)f_n( x) + f'_n (x)^2 -   {\frac{1}{2}} f'''_n(x)F_n( x) \big) \vert \psi  \vert^2.
  \end{aligned}
  \end{equation}
By Proposition \ref{Prop.regularity}, $\dot{c}$ is continuous
in the graph norm associated with $H$
and so extends continuously to the form $c$
defined again by the equation~(\ref{fin}) on $\Dom(H)$.
Then it defines a bounded map $ C \in \mathscr{B}(\Dom(H),\Dom(H)^*) $,
and the statement (d) is proved.

\subsection{The Mourre estimate}
Finally, we are concerned with the essential condition~(e).
We rewrite the restriction of~$B$ as follows
\begin{equation}\label{Bs}
\begin{aligned}
  B \upharpoonright \Dom(H) &= Hf_n +f_nH  + 2 f_n \partial^2 _y +\frac{1}{2} f_n''
  \\
  &= 2E
  +  \underbrace{(H-E)f_n +f_n(H-E)}_{B_1}
  + 2 \underbrace{f_n \partial^2 _y}_{B_2}
  + \underbrace{2E(f_n-1) + \frac{1}{2}f_n''}_{B_3}
  \,.
\end{aligned}
\end{equation}
Now, we look at the individual terms
and try to eventually estimate $P_\delta B P_\delta$
from below by a positive multiple of~$P_\delta$
plus a compact operator sandwiched between the projections~$P_\delta$s.

\subsubsection{Operator {$B_1$}}
For every $\varphi \in \sii(\Omega)$, we have
\begin{equation*}
\begin{aligned}
  | (\varphi, P_\delta B_1 P_\delta  \varphi) |
  &\leq \Vert P_\delta  \varphi \Vert^2
  \left(
  \Vert P_\delta(H-E)f_nP_\delta \Vert   + \Vert P_\delta f_n(H-E)P_\delta \Vert
  \right)
  \\
  &\leq 2 \Vert P_\delta  \varphi \Vert^2  \Vert P_\delta(H-E) \Vert
  \\
  &\leq  2 \delta \Vert P_\delta   \varphi \Vert^2
  \,.
\end{aligned}
\end{equation*}
Here we have used the spectral theorem at the last estimate.
Hence, this term can be made negligible by choosing~$\delta$ small
and we shall estimate it as
\begin{equation*}
  P_\delta B_1 P_\delta
  \geq - 2 \delta P_\delta
  \,.
\end{equation*}

\subsubsection{Operator {$B_2$}}

We demonstrate our approach on $ T^+ :=  P_\delta  f_n^+ \partial^2 _y P_\delta$;
the operator $ T^- :=  P_\delta  f_n^-\partial^2 _y P_\delta$
can be handled in a similar way.
At the same time, let us suppose that $E_l < E < E_{l+1}$.

Let~$H^+$ be the self-adjoint realization of the Laplacian in $L^2(\Omega)$,
subject to the Dirichlet boundary conditions on  $ \mathbb R  \times \{d\}$
and the Neumann boundary condition on $\mathbb R \times \{0\}$.
Let $\{\psi_k\}_{k\in\Nat^*}$ be the eigenfunctions of the one-dimensional
Laplacian in $\sii((0,d))$, subject to the Neumann boundary condition at~$0$
and the Dirichlet boundary condition at~$d$.
We define
\begin{equation*}
  \Pi_k^+\varphi(x,y)
  := \psi_k(y) \big(\psi_k,\varphi(x,\cdot)\big)_{\sii((0,d))}
  \,,
\end{equation*}
the projection on the $k$th transverse mode of~$H^+$.
We have
\begin{equation} \label{1}
  T^+ =  P_\delta ( - \sum_{k=1}^l E_k f_n^+\Pi^+_k    + R^+ )P_\delta
\end{equation}
with
\begin{equation} \label{2}
  R^+ := \sum_{k \geq l+1} -E_k  P_\delta  f_n^+ \Pi^+_k P_\delta
  \,.
\end{equation}
Note that the operator $R^+$ is not compact.
Denote by  $h_k^+ =  -\partial_x^2 \otimes 1 +  E_k$
the restriction of  $H^+ $ on  $ \Pi^+_kL^2(\Omega)$.
Let $Z :=E +i \eta$ with $\eta >0$. We have  for any $m \in \mathbb N^*$,
\begin{equation*}
  (h_k^+ -Z)^m f_n^+ \Pi^+_k = \Pi^+_k (H^+ -Z)^mf_n^+
  =  \Pi^+_k (H -Z)^m f_n^+
\end{equation*}
on the domain of the right-hand side.
Now, let us choose $ \eta := \delta  $.
If $k\geq l+1$, then
\begin{equation*}
  \Vert(h_k^+ -Z)^{-m}  \Pi^+_k \Vert
  \leq  (E_k- E)^{-m} \,.
\end{equation*}
At the same time, if $ k \geq l+1 $, we  have
\begin{align} \label{4}
  P_\delta \Pi^+_k f_n^+  P_\delta
  & = P_\delta (h_k^+ -Z)^{-m}  \Pi^+_k  (H -Z)^m  f_n^+  P_\delta
  \\
  &=  P_\delta (h_k^+ -Z)^{-m}  \Pi^+_k  f_n^+ (H -Z)^m  P_\delta
  +  P_\delta (h_k^+ -Z)^{-m}  \Pi^+_k   [ (H -Z)^m , f_n^+] P_\delta
  \,.
  \nonumber
\end{align}

The first term on the right-hand side of the second line of~\eqref{4}
can be estimated as
\begin{eqnarray} \label{5}
  \Vert P_\delta (h_k^+ -Z)^{-m}  \Pi^+_k f_n^+ (H -Z)^m  P_\delta \Vert
  \leq C (E_k- E)^{-m} \delta^m.
\end{eqnarray}
Hereafter, $C$ denotes a generic strictly positive constant
which does not depend on the index~$k$ and on~$\delta$ (but depends on fixed $E_{l+1}-E$)
and can change its value from line to line.
If $m\geq 2$,
we have
\begin{equation} \label{6}
\begin{aligned}
  \bigg\|
  \sum_{k \geq l+1} E_k P_\delta (h_k^+ -Z)^{-m}  \Pi^+_k  f_n^+ (H -Z)^m  P_\delta
  \bigg\|
  &\leq C
  \sum_{k \geq l+1 } E_k (E_k- E)^{-m}  \delta^m
  \\
  &\leq C \delta^m
  \,.
\end{aligned}
\end{equation}

Now, we turn to estimating the second term on the right-hand side
of the second line of~\eqref{4}.
We choose $m:=2$. We could improve the bound to be obtained
by choosing larger~$m$ but with more complicated calculations.
On the range of $P_\delta$, we have
\begin{equation*}
[ (H -Z)^2 , f_n^+] = 2 [(H-Z),f_n^+] (H-Z) + [(H-Z),[(H-Z),f_n^+]]
\end{equation*}
with
\begin{equation*}
\begin{aligned}
\phantom{}
   [(H-Z),f_n^+]
  &= -\big(\partial_x (f_n^+)' + (f_n^+)' \partial_x\big) \,,
  \\
   [(H-Z),[(H-Z),f_n^+]]
  &= \partial^2_x (f_n^+)''+ (f_n^+)'' \partial^2_x
  + 2 \partial_x (f_n^+)''\partial_x
  \,.
\end{aligned}
\end{equation*}
Noticing that the support of the derivative of~$f_n^+$ is compact
and not intersecting $\{x=0\}$, we use Proposition~\ref{Prop.regularity}
to obtain
\begin{eqnarray} \label{7}
  \left\Vert  [ (H -Z)^2 , f_n^+] P_\delta  \right\Vert
  \leq C ( n^{-1} \delta   + n^{-2})
  \,.
\end{eqnarray}
Consequently,
\begin{eqnarray} \label{8}
  \bigg\|
  \sum_{k \geq l+1} E_k P_\delta (h_k^+ -Z)^{-2}  \Pi^+_k
  [ (H -Z)^2 , f_n^+] P_\delta
  \bigg\|
  \leq C
  ( n^{-1} \delta + n^{-2})
  \,.
\end{eqnarray}

Summing up, we have proved that
for~$\delta$  small  and $n$~large,
\begin{eqnarray} \label{9}
  \Vert  R^+  \Vert
  \leq C (  n^{-1} \delta   +  \delta^2 + n^{-2})
  \,.
\end{eqnarray}
When analyzing~$T^-$, we consider~$H^-$ which is defined in the same
manner as~$H^+$ but with interchanged boundary conditions.
The corresponding projections~$\Pi_k^-$ and the operator~$R^-$
are defined with an obvious modification of the formulas above.
By using the same arguments as above,
we get the same estimate~\eqref{9} for~$R^-$.
Writing $R:= R^+  +R^-$,
\begin{equation*}
  T^+  +T^-
  = P_\delta  \left(- \sum_{k=1}^lE_k( \Pi^+_k f_n^+  +  \Pi^-_kf_n^-) + R\right) P_\delta.
\end{equation*}
However, since $ [ \Pi^{ \pm}_k, f_n^{ \pm}]=0$,  $\sum_{k=1}^lE_k \Pi^{\pm}_k f_n ^{\pm} \leq E_l  (\sum_{k=1}^l \Pi^{\pm}_k) f_n^{\pm}$
$\leq E_lf_n^{\pm}$,
 and we conclude with the estimate
\begin{equation} \label{10}
  P_\delta B_2 P_\delta \geq - P_\delta  (E_l
  + C \big( n^{-1} \delta +  \delta^2 + n^{-2})
  \big)P_\delta
\end{equation}
being valid in the form sense.

\subsubsection{Operator {$B_3$}}
The operator $P_\delta B_3 P_\delta$  is not small.
However, since the function $g_n:=2E(f_n-1) + \frac{1}{2}f_n''$ has a compact support,
it follows that $g_n H^{-1}$ is a compact operator.
This is seen form the fact that ${\mathsf R}(g_n H^{-1}) \subset H^1((-2n,2n) \times (0,d))$ which is compactly embedded in $L^2((-2n,2n) \times (0,d))$
by the Rellich-Kondra\-chov theorem
(see \cite[Thm.~6.2]{Adams}).
Now
\begin{equation*}
  K_n := P_\delta B_3 P_\delta = P_\delta B_3 H^{-1} H P_\delta
\end{equation*}
is also a compact operator.
Note that  the presence of~$B_3$ (its part $f_n-1$) in~\eqref{Bs}
is the  only obstruction to get a \emph{strict} Mourre estimate (i.e., with $K=0$).

\subsubsection{Conclusion}
If $E_l < E < E_{l+1}$, it follows from the preceding subsections that
for~$\delta$ small and~$n$ large,
the Mourre estimate
\begin{equation} \label{final}
  P_\delta i[H,A] P_\delta
  \geq P_\delta  \big(  2(E- E_l - \delta)
  -C (  n^{-1} \delta   +  \delta^2 +
   n^{-2} ) +  K_n \big )  P_\delta
\end{equation}
holds true,
where $K_n$ is a compact operator.

We have verified all the properties (a)--(e) required for
the application of the abstract theorem of~\cite{Mourre}.
Since~$\mathcal{T}$ is a discrete set,
this concludes the proof that the singular continuous spectrum of~$H$ is empty.

In fact, our result gives more information. In particular, the limiting absorption principle holds for every energy
$E\in \Real \setminus \mathcal{T}$; see \cite{Mourre} or \cite{Amrein}.
\\
\\
{\bf Acknowledgement}
\\
The second author (J.D.)
was supported by the Czech Science Foundation, Project No.\ 17-01706S
and the NPI CAS institutional support, No.\ RVO 61389005.
The last author (D.K.) was partially supported
by the Czech Science Foundation, Project No.\ 18-08835S.

\end{document}